\documentclass[copyright,noderivs]{eptcs}
\providecommand{\event}{AFL 2026} 

\usepackage{underscore}           
\usepackage[T1]{fontenc}

\usepackage{graphicx} 
\usepackage{amssymb}
\usepackage{amsmath}
\usepackage{amsthm}
\newcommand{\N}{\mathbb{N}}
\usepackage{hyperref}
\usepackage{comment}
\usepackage{mathtools}
\usepackage[basic]{complexity}
\usepackage{listings}
\usepackage{algorithm2e}

\newlength{\trackpad}
\newcommand{\trackbox}[2]{%
  \mathord{%
    \vcenter{%
      \hbox{%
        \begin{tabular}{|@{\hskip\trackpad}c@{\hskip\trackpad}|}
          \hline
          $\displaystyle #1$\\
          \hline
          $\displaystyle #2$\\
          \hline
        \end{tabular}%
      }%
    }%
  }%
}

\newcommand{\dollar}{\texttt{\$}}
\newcommand{\border}{\texttt{\#}}
\newcommand{\rightend}{\mathord{\vartriangleleft}}

\newcommand{\ca}{\textsf{CA}}
\newcommand{\oca}{\textsf{OCA}}
\newcommand{\range}{\textsf{range}}
\newcommand{\width}{\textsf{width}}

\newcommand{\eoe}{\ifmmode$\hspace*{\fill}$\blacksquare\else\hspace*{\fill}$\blacksquare$\fi\smallskip}

\newenvironment{problem}[1]
  {\begin{quote}\noindent\textbf{\itshape#1}\par\medskip\itshape}
  {\end{quote}}

\theoremstyle{plain}
\newtheorem{theorem}{Theorem}
\newtheorem{proposition}[theorem]{Proposition}

\newtheorem{property}[theorem]{Property}

\theoremstyle{definition}
\newtheorem{example}[theorem]{Example}

\title{Inductive Inference of Cellular Automata\thanks{Supported, in part, by Natural Sciences and Engineering Research Council of Canada Grant 2022-05092 (Ian McQuillan).
}}

\author{Martin Kutrib 
\institute{%
  Institut f\"ur Informatik, Universit\"at Giessen\\
  Arndtstr.~2, 35392 Giessen, Germany}
\email{kutrib@uni-giessen.de}
\and
Ian McQuillan
\institute{%
        Department of Computer Science, University of
        Saskatchewan\\
        110 Science Place, Saskatoon, SK S7N 5C9, Canada}
        \email{mcquillan@cs.usask.ca}
\and
Priscilla Raucci
\institute{%
  Institut f\"ur Informatik, Universit\"at Giessen\\
  Arndtstr.~2, 35392 Giessen, Germany}
\email{priscilla.raucci@uni-giessen.de}
\and
Matthias Wendlandt
\institute{%
  Institut f\"ur Informatik, Universit\"at Giessen\\
  Arndtstr.~2, 35392 Giessen, Germany}
\email{matthias.wendlandt@uni-giessen.de}
}

\def\titlerunning{Inductive Inference of Cellular Automata}
\def\authorrunning{M.~Kutrib, I.~McQuillan, P.~Raucci, M.~Wendlandt}

\begin{document}

\maketitle

\begin{abstract}
Inductive inference of one- and two-way cellular automata ($\ca$) is considered. This involves inferring a $\ca$ that is compatible with a finite amount of available data. In this paper, this information is provided in the form of a finite set of {\em intervals}, where each interval consists of two words $w$ and $w'$ over a state set alphabet, with a positive integer $i$. The goal is to infer a $\ca$ which is compatible with each interval $(w,w',i)$, meaning that it can derive $w'$ from $w$ in $i$ steps. 

We consider three variations of this problem, 1) where the $\ca$ is completely known {\it a priori}, and the goal is therefore to verify compatibility, 2) where the $\ca$ is partially known {\it a priori} and the goal is to extend it to a full $\ca$ that is compatible, and 3) where the $\ca$ is completely unknown, and the goal is to fully construct one that is compatible if one exists.
With all three variations, inference can be completed in polynomial time, and is in fact $\P$-complete.
\end{abstract}

\section{Introduction}
In the past, the majority of research in automata and formal language theory
has been studied in a bottom-up fashion, where researchers would build and
study automata or grammars to accomplish desired tasks. Grammatical inference,
or inductive inference, involves a more top-down approach, where the task is
unknown but a finite amount of data is provided as input, and the goal is to
learn or infer a grammar or automaton model, that is somehow compatible with
the input. 

Inductive inference has been studied especially for regular and context-free
languages \cite{delaHiguera2010}. For other more complex families of 
languages, less work has been done. Lindenmayer systems (L-systems) are grammatical systems where all letters of a
sentential form are rewritten in parallel somewhat similarly to cellular automata.
Inductive inference of L-systems was studied for certain problems
from the perspective of decidability already in the 1970's by Herman and
Rozenberg \cite{hermanrozenberg}. Computational complexity has recently been
studied for the problems of taking a sequence of strings as input, and to
infer an L-system that initially generates the input sequence at the beginning
of its developmental sequence. This problem is $\NP$-complete already for
deterministic context-free L-systems \cite{ComplexityInference}, but is in
$\P$ for deterministic context-sensitive L-systems if alphabets and context
sizes are fixed \cite{UCNC2018}.

Here, we study inductive inference of one-dimensional two-way cellular automata ($\ca$), and
one-dimensional one-way cellular automata ($\oca$). 
Since the forties of the last century, cellular automata have been studied and investigated 
from many different perspectives. In particular, there are numerous results concerning
their theoretical properties. Surveys on such aspects and detailed references 
are given in~\cite{Delorme:1999:capm,Kari:2005:tocas,Rozenberg:2012:HandbookNC}. 
In particular, we refer to surveys concerning computational
aspects~\cite{kutrib:2008:ca-cpv}, formal language aspects~\cite{kutrib:2009:calt}, 
and the descriptional complexity of cellular automata~\cite{kutrib:2018:cadcad}. 

A cellular automaton is a linear array of identical deterministic finite automata,
called cells. The total number of cells in the array is determined by the
input data. All cells fetch their input symbol 
during a pre-initial step.
Both $\ca$ and $\oca$ are deterministic machines that employ a local and
parallel state change function. Each cell except the two outermost ones 
of a $\ca$ is connected to both its nearest neighbors, thus,
employing a two-way communication. Each cell of an $\oca$ is connected to its
nearest neighbor to the right. So, $\oca$ employ a one-way communication from
right to left.

In contrast to the aforementioned study with L-systems, we study inductive
inference where the data includes a finite set of {\em intervals}. Each
interval consists of a triple $(w, w', i)$ where $w$ and $w'$ are strings over
some alphabet, and $i$ is a positive integer called the {\em distance}. A
$\ca$ $M$ is {\em compatible} with such an interval if $w$ and $w'$ are
sequences of states of the $\ca$, and $M$ can derive $w'$ from $w$ in exactly
$i$ steps. The goal is to output a $\ca$ that is compatible with every input
interval, if one exists. Hence, we do not necessarily have consecutive strings
derived, and they can have arbitrary distances between pairs of strings. This
is a more general problem than the one considered with L-systems in
\cite{ComplexityInference,UCNC2018}, as we can easily turn a sequence of
initial words into a finite set of intervals where all distances are exactly one. 

In this paper, we study variations of the problem. First, we consider a
version where it is simply verifying correctness, and so there
is a finite interval set and a $\ca$ $M$  as input, and the goal is to
determine if~$M$ is compatible with the set (Interval Verification Problem). 
We show this is in $\P$
(polynomial time), and is $\P$-complete. Further, it is in $\L$ (logarithmic
space) if all distances are fixed. A related problem studied in the literature
is the prediction problem; that is, given a transition function of a cellular
automaton, an initial configuration, a cell~$i$, a state~$q$, and a positive integer~$t$
in unary, is cell~$i$ in state~$q$ at time~$t$? Neary and Woods showed the
$\P$-completeness of the problem~\cite{neary:2006:pccar110}. Remarkably, they
showed the result for a fixed transition function (called rule 110) of a
two-way cellular automaton having only two states. This is the strongest
result with respect to the prediction problem obtained. This prediction problem
is related but different from the verification problem. The first difference
is that the prediction problem is for cellular automata on unbounded
configurations. Though, the time bound $t$ and cell index~$i$ imply a 
finite range of the configuration, the boundary state and its impact are missing. So, for a
verification problem with large $t$ and small initial configuration, the space
may be too small for a simulation. Another problem arises since the prediction
problem concerns one cell $i$ only. An ad-hoc application of the prediction problem
for the solution of the verification problem would require exponentially many 
state intervals to verify each bit. A further difference is that here we also
want to solve the problem for one-way cellular automata.
Apart from that, the approaches are
similar. The membership in $\P$, for example, is shown by simulating the
cellular automaton by a Turing machine more or less directly.

The next problem we consider is the inductive inference
problem, where the input is only the interval set, and the goal is to
determine if there is a $\ca$~$M$ compatible with the set, and if so, to
construct one. Similarly, we show here that this problem is in $\P$, and is
$\P$-complete. Finally, we consider the problem where potentially a subset of
transitions of a $\ca$ is provided as input along with an interval set, and
the goal is to extend the transitions into a full $\ca$ that is compatible
with the set if it exists. We again show that this problem is in $\P$, and is
$\P$-complete.

\section{Preliminaries}\label{sec:prelim}

We denote by $\Sigma^*$ the set of all words on the finite alphabet $\Sigma$, including the empty word $\lambda$, and let
$\Sigma^+ = \Sigma^* \setminus \{\lambda\}$. For any word $w\in\Sigma^*$, we let 
$|w|$ denote its length, $w^R$ its reversal.
If $w = xyz, x,y,z \in \Sigma^*$, then 
$x$ is a \emph{prefix} of $w$, $y$ is a \emph{subword} of $w$ and $z$ is a \emph{suffix} of $w$.
We use $\subseteq$ for {inclusions}, and $\subset$ for proper inclusion. 
A language over $\Sigma$ is any subset of $\Sigma^*$. 
Given a set $S$, we denote its cardinality by
$|S|$. 

A (one-dimensional) cellular automaton is a linear 
array of identical deterministic finite automata,
called cells, numbered $1,2,\dots,n$. 
The local state transition depends on the current state of a cell itself
and the current states of its neighbors, where the outermost cells 
receive information associated with a boundary symbol 
on their free input lines. The cells work synchronously at discrete time
steps.

Formally, a \emph{two-way cellular automaton} $(\ca)$ is a system 
$M = \langle Q, F, \Sigma, \border, \delta \rangle$, where
$Q$ is the finite, nonempty set of \emph{cell states},
$F\subseteq Q$ is the set of \emph{accepting states},
$\Sigma\subseteq Q$ is the finite, nonempty set of \emph{input symbols},
$\border \not\in Q$ is the \emph{boundary state}, and
\mbox{$\delta\colon (Q \cup\{\border\}) \times Q \times (Q\cup\{\border\}) \to Q$}
is the \emph{local transition function} (see Figure~\ref{fig:occa}).
We only require here for $\delta$ to be a partial function rather than a  total function, in contrast to most definitions in
the literature. We do this because with the goal of inferring $\ca$, we only need to infer transitions needed on the input data.
Clearly, if we infer a partial function, it is very easy to extend this to a total function by adding $\delta(x,y,z) = y$ (or any arbitrary
target) for any $x,y,z$ where the partial function is not defined.

A \emph{one-way cellular automaton} $(\oca)$ is a cellular automaton in
which each cell receives information from its immediate neighbor to the 
right only. So, the flow of information is restricted to be from right to left.
Formally, $\delta$ is a mapping from $Q \times (Q\cup \{\border\})$ to~$Q$ 
(see Figure~\ref{fig:occa}).

\begin{figure}[!b]
\centering
\includegraphics[scale=1]{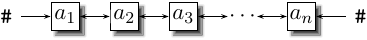}\\[3mm]
\includegraphics[scale=1]{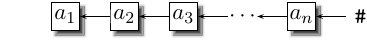}
\caption{A two-way (top) and a one-way (bottom) cellular automaton.}
\label{fig:occa}
\end{figure}

A \emph{configuration} $c_t$ of $M$ at time $t\geq 0$ is a
description of its global state, which is formally a
mapping \mbox{$c_t:\{1,2,\dots,n\} \to Q$,} for~\mbox{$n\geq 1$.} 
The configuration at time 0 is defined by the given input 
\mbox{$w=a_1a_2\cdots a_n\in \Sigma^+$.} We set 
$c_0(i)=a_i$, for $1\leq i\leq n$. 
Configurations may be represented as words over the set of cell states
in their natural ordering. For example, the initial configuration
for $w$ is represented by $a_1 a_2 \cdots a_n$.
Successor configurations are computed
according to the global transition function~$\Delta$, that is,
\mbox{$c_{t+1}=\Delta(c_t)$},
as follows.

For $\ca$ we let  
\begin{eqnarray*}
 c_{t+1}(1)&=&\delta(\border,c_t(1),c_t(2))\\
 c_{t+1}(i)&=&\delta(c_t(i-1),c_t(i),c_t(i+1)), \mbox{ for }i \in \{2,3, \ldots, n-1\}\\
 c_{t+1}(n)&=&\delta(c_t(n-1),c_t(n),\border)
\end{eqnarray*}
if $n\geq 1$ and, for $n=1$, the next state of the sole cell is $\delta(\border, c_t(1),
\border)$.

For $\oca$ we let  
\begin{eqnarray*}
 c_{t+1}(i)&=&\delta(c_t(i),c_t(i+1)), \mbox{ for }i \in \{1,2, \ldots, n-1\}\\
 c_{t+1}(n)&=&\delta(c_t(n),\border).
\end{eqnarray*}

Thus, the global transition function~$\Delta$ is induced by $\delta$.

We index matrices starting at row $0$ and column $0$.
The computation of $m-1$ steps of a $\ca$ on an input $a_1 \cdots a_n$ can be visualized as a \emph{space-time diagram}, which is an $m  \times (n+2)$ matrix, where row $i$ contains an encoding of $ \border c_i \border $,
with $(i,0)$ and $(i,n+1)$ containing $\border$, and $(i,j)$ containing~$c_i(j)$ otherwise.

Sometimes in the sequel, we consider cellular automata as acceptors for 
formal languages; that is, due to their space and time bounds they 
can be seen as deciders. They decide whether a given input word belongs
to a formal language or not. 
A cellular automaton $M$ \emph{accepts} an input $a_1a_2\cdots a_n\in \Sigma^+$,
if at some time during the course of its computation
the leftmost cell enters an
accepting state. 
The \emph{language accepted by~$M$} is
\mbox{$L(M)= \{\,w\in \Sigma^+\mid w \text{ is accepted by } M\,\}$.}

In the following we consider several problems in connection with inductive
inference. To this end, we introduce some notations.
Given a state set $Q$
(interpreted as an alphabet), we call $(w,w', i)$, where $w, w' \in Q^+$, $|w|
= |w'|$, a \emph{state interval}, $w$ is the \emph{source} of the interval, $w'$ is the \emph{target} of the interval, and $i\geq 1$ is the \emph{distance} of the interval. We call a
non-empty finite set $\rho = \{ (w_j,w_j', i_j) \mid 1\leq j\leq m\,\}$, a \emph{state interval set} of size $m$. The \emph{width
of $\rho$} (denoted by $\width_{\rho}$) is the length of the longest word in $\rho$, that is, 
$\max\{\, |w_j|\mid 1\leq j\leq |\rho|\,\}$, its \emph{range} (denoted by
$\range_{\rho}$) is its largest distance, that
  is, $\max\{\, i_j\mid 1\leq j\leq |\rho|\,\}$. 
Note that, unless otherwise stated, the distances are given in unary. On the
one hand, this is to be consistent with the prediction problem mentioned
above, where the time is given in unary. On the other hand, this is
is assumed because we are inferring possible intermediate words in any
interval. 
Essentially, the input should contain one symbol as part of the distance for each unknown
configuration, or row of a space-time diagram.
In Section~\ref{sec:conclusions}, we briefly discuss parameterized complexity.
If the distances are given in binary, one can use the range $k$ of the state
interval set as parameter and obtains fixed-parameter tractable (FPT) solutions
by using the ``harmless'' function $f(k)=k$.
 
Hence, the {\em size} $|(w, w', i)| = |w| +  |w'| + i$ of an interval $(w,
  w', i)$ is extended to sets of intervals by summing up the sizes of the intervals.
We say that a $\ca$ (or $\oca$) $\langle Q, F, \Sigma, \border, \delta\rangle$
is \emph{compatible} with state interval set 
$(w_1, w_1', i_1), \ldots, (w_m, w_m', i_m)$ if the global transition function
$\Delta$ of the $\ca$ obeys $\Delta^{i_j}(w_j) = w_j'$, for each $1 \le j \le
m$.

An immediate observation is useful that yields the following property.

\begin{property}\label{property:radius}
Given a $\ca$ $\langle Q,  F,\Sigma, \border, \delta \rangle$, 
a basic observation is that the state of some cell 
$1\leq i\leq n$ at time $t\geq 1$ depends only on the states of the cells
\mbox{$i-t,i-t+1,\dots,i,$} \mbox{$i+1,\dots ,i+t$,} where the states outside
of $1,2,\dots ,n$ are assumed to be $\border$ and $\delta$ never replaces $\border$
by another state. This is clear by induction, as at each time step, cell $i$ depends only on $i-1, i$, and $i+1$ of the previous step. Similarly, this observation applies to $\oca$ as well.
\end{property}

\section{Interval Verification Problems}\label{sec:int-veri-prob}

Essentially, the Interval Verification Problem involves verifying that
a $\ca$ or $\oca$ can be used within certain parts of computations. 

\begin{problem}{$\ca$ (resp.\ $\oca$) Interval Verification Problem.}
Given a state interval set $\rho = (w_1, w_1', i_1), \ldots, (w_m, w_m',
i_m)$, and a $\ca$ (resp.\ $\oca$) $M$, is $M$ compatible with $\rho$? 
\end{problem}

The problem can be used as a step within the inference procedures.
Here we determine algorithms and their complexity for solving it. 
Let us start with a special case.

\begin{proposition}\label{prop:verify-constant} 
For each constant $r\geq 1$, the $\ca$ (resp.\ $\oca$) Interval Verification Problem can be decided in $\L$ if
the ranges of the state interval sets are bounded by $r$.
\end{proposition}

Next, we turn to the general case of the Interval Verification Problem, where
the ranges are not necessarily bounded by a constant. As mentioned in the
introduction, a related problem is the prediction problem, whose $\P$-completeness is shown
in~\cite{neary:2006:pccar110}. The membership in $\P$, for example, follows by
the observation that the cellular automata can be simulated by a Turing
machine in quadratic time. For the sake of completeness, here we give more
details of this construction. Since here we are also interested in the
problems for one-way cellular automata, for simplicity, we use a direct reduction of the 
Monotone Circuit Value Problem~\cite{goldschlager:1977:mpcvppcomplete}
to $\oca$.

\begin{theorem}\label{theo:verify} 
The $\ca$ (resp.\ $\oca$) Interval Verification Problem is $\P$-complete.
Both are also true on one interval.
\end{theorem}

\begin{proof}
Let $M=\langle Q,  F,\Sigma, \border, \delta \rangle$ be a $\ca$ and
$\rho = (w_1, w_1', i_1), \ldots, (w_m, w_m', i_m)$ be a state interval set.

We start by showing the containment in $\P$.
The construction of a Turing machine $T$ certifying that the problem belongs to
$\P$ is more or less straightforward. In order to test one state interval
$(w, w', i)$, the Turing machine~$T$ 
uses two working tapes. Assume that $w$ is the inscription of the first
tape. Then~$T$ computes the successor configuration $w_1=\Delta(w)$ by sliding
a window of size three over $w$ and applying $\delta$ to each window
content. The result is written on the second tape. Next, $T$ does the same
with $w_1$, whereby the result is written on the first tape. By alternating
the role of the two tapes, machine $T$ simulates the computation of $M$ on $w$. Now it
is sufficient to maintain a counter that is increased whenever a new
configuration of $M$ is simulated. The process stops when the counter value
coincides with $i$. Finally, $T$ compares the configuration computed at last
with $w'$. If they are different, the test fails and $T$ stops. If and only if
all state intervals of $\rho$ are tested successfully, $M$ is compatible with
$\rho$.

Concerning the time complexity of $T$, we can see that the simulation of one
application of $\delta$ takes at most $O(n)$ time, 
where $n$ is the input length, which is at least 
$$
|w_1|+|w_2|+ \cdots +|w_m|+\range_{\rho} + |M|.
$$
(For the simulation, $T$ has to read the window content and has to find the input
position in $M$ at which~$\delta$ is defined for it.) 
So, it takes at most $O(|w|\cdot n)$ time to simulate one transition of
$\Delta$ when state interval $(w, w', i)$ is verified. Therefore, the
verification of the state interval takes $O(|w|\cdot n\cdot i)$ time steps.
For all state intervals, we obtain at most
$$
(|w_1|+|w_2|+ \cdots +|w_m|)\cdot n\cdot \range_{\rho}
\leq
n^2\cdot \range_{\rho} \leq n^3
$$
time steps.

\medskip 

\begin{sloppypar}
Next, we will see that the problem is $\P$-hard. We use the $\P$-complete Monotone Circuit
Value Problem~\cite{goldschlager:1977:mpcvppcomplete}. The input is a directed acyclic graph $(V,E)$ with
\mbox{$V= \{1, 2,\ldots, n\}$} and all edges are of the form $(i,j), i <j$. Each vertex $k$ (called a gate)
is associated via $s(k)$ with  $\{t,f, \vee,\wedge\}$ ($t$ means true and $f$
means false). If the indegree of $k$ is $0$, then $s(i)\in \{t,f\}$, and if
the indegree is $2$, then $s(k) \in \{\vee,\wedge\}$ (there are no negated
values with the monotone version of this problem). An example of such a graph is in Figure~\ref{fig:Phard}~(a).
Then the output of vertex $k$, $T(k) \in \{t,f\}$ is $s(k)$ if $s(k) \in
\{t,f\}$; and when the indegree of $k$ is $2$ with incoming edges $i$ and $j$,
$T(k) = T(i) \wedge T(j)$ if $s(k) = \wedge$, and $T(k) = T(i) \vee T(j)$,
if $s(k) = \vee$. 
Finally, $T(n)$ is the output of the instance. 
\end{sloppypar}

From an instance of this problem, we will use a log-space reduction to
construct an $\oca$ $M$ and a state interval $(w,w', i)$ such that $M$ is
compatible with $(w,w',i)$ if and only if the Monotone Circuit Value Problem
outputs $t$. See Figure~\ref{fig:Phard} for an example.

\begin{figure}[!t]  
\centering
  \begin{tabular}{@{}cc@{}} 
    (a) & \includegraphics[scale=.7]{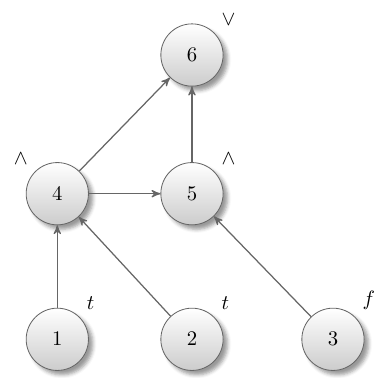}\\[1em]
    (b) & \includegraphics[scale=.8]{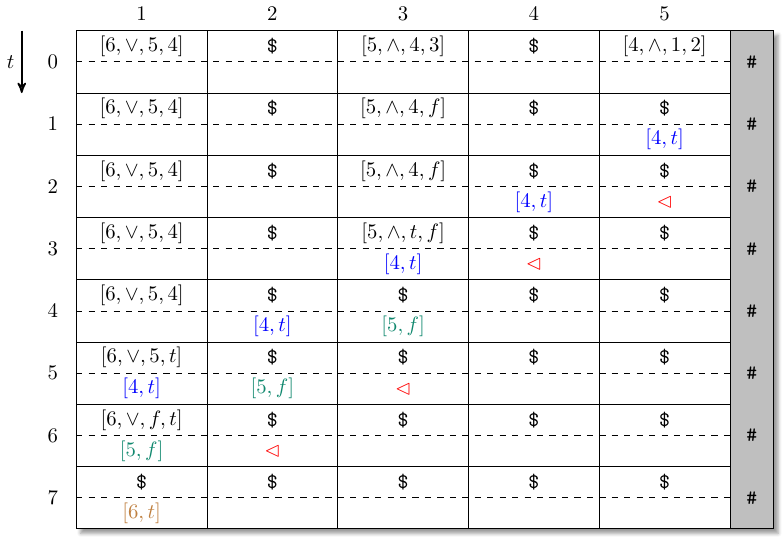}
  \end{tabular}
  \caption{A monotone circuit value graph in (a), and the simulation of the $\oca$ on
    the instance created from this circuit in (b). The information being passed one
    cell at a time to the left is shown in blue for vertex 4, green for vertex
    5, and dark yellow for vertex 6. The symbol $\rightend$ is used to show
    the ``here is the end'' signal.} 
  \label{fig:Phard}
\end{figure}

Denote~$s(k)$ by $s_k$. 
For the states, we use two tracks. Essentially, the first track describes the graph, and the second track is used for the $\oca$ to communicate truth values one cell at a time towards the left.
From the instance, construct the following string $w$: 
\begin{equation}  
\trackbox{[k_m, s_{k_m}, j_m, i_m]}{} \trackbox{\dollar}{}  \cdots \trackbox{[k_2,
    s_{k_2}, j_2, i_2]}{}  \trackbox{\dollar}{}  \trackbox{[k_1, s_{k_1}, j_1,
    i_1]}{} 
    \label{ocastart}
\end{equation}
where $1,\dots , k_1-1$ are the vertices with indegree $0$;
$k_{l+1} = k_l + 1$ for $1 \le l <m$ and $k_m =n$; 
the two incoming edges to $k_l$ are from vertices $i_l$ and $j_l$. Therefore, the first components in this string are only present for each vertex in order that has indegree $2$. The second component $s_{k_l}$ encodes the operation which must be~$\vee$ or~$\wedge$ since they all have indegree $2$, and the third and fourth component $j_l, i_l$ encode the two incoming vertices (which
are either themselves represented to the right in the string, or have indegree $0$).

The main part of the $\oca$ $M$ is independent of the instance. The dependent
part is the assignment of input truth values to the variables of the gates.
In the first step, for all $[k_l, s_{k_l}, j_l, i_l]$ with $l > 1$, $M$ replaces $i_l$ and $j_l$ in each cell $l$ if they are 
incoming edges from input gates (having indegree $0$). 
So,~$j_l$ is replaced with  $s_{j_l}$ if $j_l < k_1$. Otherwise, if $j_l\geq k_1$
then $M$ does not replace $j_l$ and all symbols are kept the same.
In a similar fashion, $M$ replaces $i_l$ with $s_{i_l}$.
However, in the first step, the rightmost cell containing $[k_1, s_{k_1},j_1, i_1]$ can identify itself with the
help of the boundary symbol to its right. Its variables can always be
replaced. So, in the first step, $M$ computes $T({k_1})$
and puts $[k_1, T({k_1})]$ in the second track.
Simultaneously, a $\dollar$ replaces what is written on the first track (which
indicates that it is now done with this cell). 

In the following, the cells work as follows. Unless otherwise stated,
the input of the second track is shifted to the left one symbol at a time 
at each step. The rightmost cell that has written a symbol of the form $[l, T({l})]$ onto
the second track next writes a symbol $\rightend$ onto the second track.
This symbol indicates a leftmoving ``here is the end'' signal. It is used
to ensure that cells write to the second track only at the end of the
leftmoving information. So, the $\rightend$ appears after two steps in the
rightmost cell for the first time.

When a symbol of the form $[k, T({k})]$ on the second track is moved to a cell
containing $[k_l, s_{k_l}, j_l, i_l ]$ in the first track,
where $i_l=k$ or $j_l=k$,
then the corresponding number in the first track 
is replaced with~$T(k)$ in the same step. If a cell $l$ sees the $\rightend$
in its right neighbor, its $i_l$ and $j_l$ values have necessarily been replaced by
$t$ or $f$. So, in the next step, the cell  
computes $T({l})$ (by applying $s_{l}$ to the now known truth
assignment of the incoming edges), and puts $[l, T({l})]$ in the second
track, whereby the left moving~$\rightend$ is overwritten. As before,
simultaneously, a $\dollar$ is written on the first track.
The symbol $\rightend$ is written onto the second track in the next step,
now indicating the new end of the leftmoving information.
This is the end of the construction of $M$. The parts that depend on the
instance can be done in log-space.

See Figure~\ref{fig:Phard} for a visualization of this simulation.

The second component $w'$ of the state interval $(w,w',i)$ is now
constructed as $\trackbox{\dollar}{[n,t]} \trackbox{\dollar}{}^{2m-2}$,
where $n$ is the number of gates with indegree two.
Finally, the distance $i$ of the state interval is set to
$3m-2$.

In this fashion, in the first step $j_1$ and $i_1$ must be known, so they are
replaced and $T(k_1)$ is computed.
Then, the shifting on the second track starts, and $\rightend$ is written onto the second track.

When $T(k_1)$ is passed leftward, this necessarily provides enough information to
compute $T(k_2)$, which can then be passed leftward etc.
It is clear that  the instance will be true
if and only if $M$ reaches the state $\trackbox{\dollar}{[n,t]}$ in the leftmost
cell. This can happen only if the $\rightend$ has been seen by the cell in this
last step. Whenever the $\rightend$ leaves a cell then a $\dollar$ is on
the first track. So, the instance will be true
if and only if $M$ reaches the configuration
$\trackbox{\dollar}{[n,t]} \trackbox{\dollar}{}^{2m-2}$. 
The signal $\rightend$ is established in the rightmost cell at time two.
Then it would need further $2m-2$ steps to reach the leftmost cell. But on its way it
is delayed by one step at each cell whose input represents a gate except for
the leftmost one. Altogether this makes 
$i=2m-2+2+m-2=3m-2$ time steps.

We conclude that $M$ is compatible with $\rho=(w,w', 3m-2)$ if and only if the 
instance of the Monotone Circuit Value Problem is true. 

Also, the proof implies the $\P$-completeness for $\ca$. 
\end{proof}

\section{Interval Inference Problems}\label{sec:int-inference-probs}

In this section, we turn to two problems that involve determining whether
there is a $\ca$ (resp.\ $\oca$) that could be used within certain parts of
computations. That is, given a state interval set $\rho$,
is there a $\ca$ (resp.\ $\oca$) that is compatible with~$\rho$? In contrast
to the Interval Verification Problem, now the cellular automaton is not
given. Moreover, we require that if there is a solution, then 
the procedure also must construct one such cellular automaton.

\newpage

\begin{problem}{$\ca$ (resp.\ $\oca$) Interval Inference Problem.} 
Given a state interval set $\rho = (w_1, w_1', i_1), \ldots, (w_m, w_m',
i_m)$, is there a $\ca$ (resp.\ $\oca$) that is compatible with $\rho$? If yes
then construct it.
\end{problem}

We now prove the main result.
\begin{theorem}\label{theo:interval-inference-prob-new}
The $\ca$ (resp.\ $\oca$) Interval Inference Problem is in $\P$. 
\end{theorem}

\begin{proof}
Let $\rho = (w_1, w_1', i_1), \ldots, (w_m, w_m', i_m)$ be a given state
interval set such that $i_1\leq i_2 \leq \cdots \leq i_m$.
For $1\leq j\leq m$, let $I_j$ denote $(w_j, w_j', i_j)$, and 
for some $n_j\geq 1$, let 
$$
w_j=a_{j,1}a_{j,2}\cdots a_{j,n_j} \text{ and } 
w_j'=b_{j,1}b_{j,2}\cdots b_{j,n_j}.
$$
For easier writing
we additionally set $a_{j,p}=b_{j,p}=\border$ if $p \in \mathbb{Z}$ and $p$ is either less than $1$ or more than $n_j$.

First, we present an algorithm that constructs a $\ca$
$M=\langle Q, \Sigma, \border, \delta \rangle$ (we leave off the final state set as it is not relevant)
compatible with $\rho$ if possible, and halts with a failure message
otherwise (see Example~\ref{exa:it-inv-prob} below) along with a
proof of correctness. We conclude with an analysis of the complexity. 

Let $\Sigma $ be defined to contain all states appearing in the state
intervals with $k = |\Sigma|$; that is, all
symbols occurring in $w_j$ and $w_j'$, for any $1\leq j\leq m$.
For the purposes of the procedure below, we assume without loss of generality that
$\Sigma$ consists of consecutive natural numbers from $1$ to some $k$, and we
replace $\border$ in the procedure with the number $0$ (we still write
$\border$ below for readability but it is always replaced by $0$ in the
algorithm).  
The most critical part is the construction of the state set $Q$ (which we
similarly assume are all natural numbers), and  $\delta$.

The procedure will make a list $s$ that will ultimately encode the transitions
and the list of states, indexed by the state itself. The contents of 
position~$i$ of~$s$ will be a set of triples of states that have incoming transitions
to state $i$, and so $(x,y,z) \in s[i]$ implies that 
$\delta(x,y,z) = i$ is a transition. If the $\ca$ is to be compatible with
$\rho$, it is clear that each $(x,y,z)$ that is present in $s$ should only 
appear once in $s$ due to determinism, but in principle there can be multiple
incoming transitions to a single state and so each position of $s$ contains a
list of triples. 
With our particular construction, only the states of $\Sigma$ will have possibly multiple incoming transitions.

Basically, the idea of the algorithm is to start the construction
with the state interval $I_1$ having the shortest 
distance and continue with the other state intervals in increasing order of their
distances. After having successfully processed some state interval~$I_j$,
the constructed cellular automaton is compatible with state intervals
$I_1$ to $I_j$. So, for $j=m$ the algorithm stops.

The main procedure in Algorithm \ref{mainalg} is divided into a main part (lines~\ref{initsbegin}--\ref{endofmain})
plus an additional sub-routine called \texttt{update} (lines~\ref{startofupdate}--\ref{endofupdate}). 

In lines~\ref{initsbegin}--\ref{endinitsbegin} below, the list $s$ is initialized
and set to be indexed by the alphabet $\Sigma\cup \{\border\} = \{0,
\ldots, k\}$, with an empty list at each position. For later use, a list $upd$
is initialized to by empty. The purpose of $upd$ is to store tasks of state
replacements
to resolve temporary conflicts during construction. Each element of $upd$ is
itself a list consisting of at least two states, where all but the last one
have to be identified and replaced by the last one.

The main part loops from $1$ to $m$ over the state intervals
(line~\ref{mainloop}), where in each loop the next state interval is added.
Assume now that interval $I_{j-1}$ (or nothing if $j=1$) has successfully
been processed and the cellular automaton constructed so far is compatible 
with the state intervals $I_1$ to $I_{j-1}$.

Next, the interval $I_j$ is added. To this end, a matrix $N_j$ 
of size $(i_j+1) \times (n_j+2)$ is constructed (line~\ref{mainnewmatrixbegin}), 
where each position of the matrix will contain states. The first row  of $N_j$ is initialized
to be the first word $\border w_j \border$ of the interval in line~\ref{setfirstrow}.
The part from line~\ref{setotherrows} to~\ref{endfillrows} fills the other
positions of the matrix. This is done from row~$1$ to row~$i_j$ by applying
the rules of the already constructed transition function from left to right, if possible. If 
a rule is still missing, a new state is used and the rule is stored in the
list $s$.
The matrix will ultimately contain the space-time diagram of interval $I_j$,
but at the end of line~\ref{endfillrows}, it does not yet take into account the
target word $w'_j$. Apart from that, the cellular automaton constructed so far 
is compatible with all matrices $N_1$ to $N_j$.

The target word $w'_j$ has next to replace the current last row (row $i_j$) of
$N_j$. This happens from position~$1$ to position $n_j$ in the loop
beginning in line~\ref{startlastrow}.
First, the current state of a position $c$ is stored into $i$. Then it is tested
if there is actually something to be replaced; that is, if $i\neq
b_{j,c}$. If in addition~$i$ is some state from $\Sigma$ (that is $i
\leq k$), then the construction fails and the algorithm stops with a failure
message, since a state from $\Sigma$ must be
replaced by a different state from~$\Sigma$ (see correctness below).
Otherwise, state~$i$ has to be replaced by $b_{j,c}$, and this replacement has
to be done everywhere in the already constructed matrices~$N_1$ to~$N_{j}$
and transition rules as well. To prepare for that,
the list $(i, b_{j,c})$ is pushed to $upd$, now
being the only entry in $upd$ (and so $upd$ is a list of lists).
In general, the sub-routine \texttt{update} processes all tasks in the
list~$upd$.

\noindent
\begin{algorithm}[!t]
\lstset{escapeinside={(*@}{@*)}}
\begin{lstlisting}
create list $s := \emptyset$ (*@\label{initsbegin}@*)
for each $e$ from $0$ to $k$
   push new empty list to $s[e]$ (*@\label{endinitsbegin}@*)
create new list $upd := \emptyset$ 

for each $j$ from $1$ to $m$              (*@\label{mainloop}@*)
   create matrix $N_j$ of size $(i_j+1) \times(n_j+2)$ (*@\label{mainnewmatrixbegin}@*)
   for each $c$ from $0$ to $n_j+1$
      $N_j(0,c) := a_{j,c}$ (*@\label{setfirstrow}@*)
   for each $l$ from $1$ to $i_j$ (*@\label{setotherrows}@*)
      $N_j(l,0) := 0, N_j(l,n_j+1) := 0$        //$0$ is $\border$
      for each $c$ from $1$ to $n_j$ (*@\label{beginfillrows}@*)
         $x := N_j(l-1, c-1)$
         $y := N_j(l-1, c)$
         $z := N_j(l-1, c+1)$
         if $(x,y,z)$ is already stored in $s$  (*@\label{fillrowsif}@*)
            let $i$ be its index
         else 
            push new empty list in $s$ at new last index $i$
            push $(x,y,z)$ to $s[i]$
         $N_j[l,c] := i$ (*@\label{endfillrows}@*)    (*@\label{mainnewmatrixend}@*)

   for each $c$ from $1$ to $n_j$ (*@\label{startlastrow}@*)
      $i :=  N_j[i_j,c]$  
      if $i\neq b_{j,c}$ and $i \leq k$ (*@\label{reject1}@*)
         print `reject'
      else 
         push list $(i,b_{j,c})$ to $upd$
         update($s$,$N_1, \ldots, N_j$,$upd$) (*@\label{endofmain}@*)(*@\label{endlastrow}@*)
 \end{lstlisting}
 \label{mainalg}
 \caption{Main Algorithm for Inference.}
 \end{algorithm}

The purpose of the sub-routine \texttt{update} is to resolve conflicts after some
state~$i$ in the last row of a matrix has to be replaced by another symbol
from $\Sigma$. As mentioned before, the purpose of $upd$ is to store tasks of state
replacements. Each element of $upd$ is itself a list whose last element is the new
state that replaces all occurrences of all other states in that list in all matrices
$N_1$ to $N_j$.
However, this is not sufficient since also all occurrences of those states
in the triples of transition rules must be replaced. This can cause
some `new' triples to now act nondeterministically because they appear at more than
one position in $s$, thus, mapping to different states. To resolve these
problems, the sub-routine \texttt{update} stores new tasks into $upd$. 
When \texttt{update} is called at line
\ref{endofmain} of the main procedure above, 
the list $upd$ contains only one task consisting of the current symbol (of a last row) and the value of
its replacement. A temporary list $tmp$ is initialized.

At line~\ref{startofupdateloop} of \texttt{update}, a loop is established that
is repeated as long as list $upd$ is non-empty; that is, as long as there is
something to be replaced. 
In line~\ref{updategetjob} the next task from $upd$ is popped to the auxiliary
list $tmp$, and in line~\ref{updategetdest} the new state is popped from $tmp$
and stored in the variable $dest$.
Then the process visits all positions in the
matrices~$N_1$ to~$N_j$ and replaces all occurrences of all states in the list
$tmp$ by the value of $dest$
(lines~\ref{updateallcellsbegin}--\ref{updateonecellsend}).
Moreover, whenever a replacement takes place, the corresponding triple that
maps to the old state (triple $(x,y,z$)) is moved unchanged in $s$ from the list indexed by the
old state to the list indexed by the new state. So far, all replacements
of the current loop have been done in the matrices $N_1$ to $N_j$. Next,
the states in $tmp$ have to be replaced in all triples of the transition
function as well (lines~\ref{updatetransitionsbegin}--\ref{updatetransitionsend}).
From now on, all states from $tmp$ no longer appear anywhere in the matrices and
in the transition function. So, we can forget them and set $tmp$ to be empty.

The rest of the current loop of sub-routine \texttt{update} prepares to resolve
nondeterministic transitions that may have been introduced before by creating
and storing new tasks into $upd$. For example,
if $(x,y,i)\in s[p]$ and $(x,y,z)\in s[q]$, and $i$ is replaced by $z$, then
we have the conflict of $(x,y,z)\in s[p]$ and $(x,y,z)\in s[q]$.
The loop starts at line~\ref{updateresolveconflictsbegin} with empty list
$tmp$. It examines~$s$ for triples $(x,y,z)$ occurring in
more than one entry of $s$ and processes each such triple in a pass through
the loop. If no such triple is found, the loop and, hence, the current pass through
the while-loop ends. In order to resolve the nondeterminism, exactly one state $t$ has to
be determined to which $(x,y,z)$ maps. Therefore, for each entry~$\ell$ of~$s$ 
containing $t$, $\ell$ has to be replaced and is added to list $tmp$.
Then $(x,y,z)$ is removed from list~$s[\ell]$ 
(lines~\ref{updateresolvecollectconflicts}--\ref{updateresolvecollectconflictstwo}).
Now, state $t$ has to be determined 
\mbox{(lines~\ref{updateresolveclassifyconflictsend}--\ref{updateresolvecollectconflictsend}).}
If $tmp$ contains two different states from~$\Sigma$, the algorithm stops with a failure
message, since two different states from~$\Sigma$ must be replaced by one new
state. 
If $tmp$ contains exactly one state from~$\Sigma$ this must be~$t$, since
otherwise some state from $\Sigma$ is replaced. Finally,
if $tmp$ does not contain a state from $\Sigma$, state $t$ is set to be a new
state and its position in $s$ is initialized. 
Now, it remains to add $(x,y,z)$ to the list $s[t]$, to push~$t$ to $tmp$ (now
$tmp$ contains a complete task for $upd$), and to push $tmp$ to $upd$. Finally,
the list $tmp$ is cleared and the next pass through the while-loop 
at line~\ref{startofupdateloop} begins.

\noindent
\edef\mynextline{\number\numexpr\getrefnumber{endofmain}+1\relax}
\begin{algorithm}[!t]
\lstset{escapeinside={(*@}{@*)}}
\begin{lstlisting}[firstnumber=\mynextline]
update($s$,$N_1,\ldots, N_j$,$upd$){  (*@\label{startofupdate}@*)
   create new list $tmp := \emptyset$   (*@\label{initsend}@*)
   while list $upd$ is not empty (*@\label{startofupdateloop}@*)
      pop $tmp$ from $upd$      (*@\label{updategetjob}@*)
      pop $dest$ from $tmp$     (*@\label{updategetdest}@*)
      for each $h$ from $1$ to $j$ (*@\label{updateallcellsbegin}@*)
         for each $l$ from $1$ to $i_h$
            for each $c$ from $1$ to $n_h$ (*@\label{updateallcellsend}@*)
               if $tmp$ contains entry $N_h[l,c]$ and $N_h[l,c]\neq dest$  (*@\label{updateonecellsbegin}@*)
                  $x := N_h(l-1, c-1)$
                  $y := N_h(l-1, c)$
                  $z := N_h(l-1, c+1)$
                  pop $(x,y,z)$ from $s[N_h[l,c]]$
                  push $(x,y,z)$ to $s[dest]$
                  $N_h[l,c] := dest$       (*@\label{updateonecellsend}@*)
 
      for each entry $i$ of $tmp$        (*@\label{updatetransitionsbegin}@*)
         for each occurrence of $i$ in any entry $\ell$ of $s$ 
            replace $i$ with $dest$ in $s[\ell]$ if not already in it (*@\label{updatetransitionsend}@*)
      $tmp := \emptyset$

      for each $(x,y,z)$ occurring in more than one entry of $s$ (*@\label{updateresolveconflictsbegin}@*)
         for each entry $\ell$ containing $(x,y,z)$   (*@\label{updateresolvecollectconflicts}@*)
            push $\ell$ to $tmp$
            pop $(x,y,z)$ from $s[\ell]$    (*@\label{updateresolvecollectconflictstwo}@*)
         if $tmp$ contains entries $p\neq q$ with $p<q\leq k$ (*@\label{updateresolveclassifyconflictsend}@*)
            print `reject'
         else if $tmp$ contains an entry $q\leq k$
            $t := q$
         else
            push new empty list in $s$ at new last index $t$ (*@\label{updateresolvecollectconflictsend}@*)
         push $(x,y,z)$ to $s[t]$                     
         push $t$ to $tmp$ 
         push $tmp$ to $upd$
         $tmp := \emptyset$ (*@\label{endofupdate}@*)
} 
 \end{lstlisting}
\caption{The update sub-routine.}
\end{algorithm}

We turn to give evidence that the algorithm is correct. First, we see that it
terminates for all inputs as follows. For each input, in the main part there are only
for-loops, each with a fixed number of iterations. Additionally, there is a
call of \texttt{update}. The sub-routine \texttt{update} contains a while-loop
that is passed through as long as the list $upd$ is not empty. For each call,
in the first iteration the list contains exactly one entry. Subsequently,
either \texttt{update} terminates, or it contains at least one task in which
two states have to be replaced. So,
at least two states in the matrices are replaced by one. So, in each
further iteration, the total number of states in the matrices is strictly
decreasing. This implies that \texttt{update} and, hence, the algorithm halts
on any input.

Next, we consider situations in which the algorithm halts with a failure
message. This may happen in line~\ref{reject1} or
line~\ref{updateresolveclassifyconflictsend} whenever some letter from
$\Sigma$ is to be replaced by another letter. In these cases, there does not
exist a cellular automaton compatible with $\rho$.
Why? Assume that interval $I_{j-1}$ (or nothing if $j=1$) has successfully
been processed and the cellular automaton constructed so far is compatible 
with the state intervals $I_1$ to $I_{j-1}$. Then, matrix $N_j$ is processed
next. This is done in the main loop of the main part up to 
line~\ref{mainnewmatrixend}, though the last line will not be included at
first. Up to this point, the cellular automaton constructed is compatible
with $I_1$ to $I_{j-1}$ and with the current version of $I_j$.
Then from left to right the states in the last row of $I_j$ are replaced
by the letters from $w_j'$. In any cellular automaton compatible with~$\rho$,
these letters must appear at their positions. Now, by
Property~\ref{property:radius}, we derive that, in particular, each
state from $\Sigma$ at row $l$ and column $c$ depends only
on the states of the cells \mbox{$c-i,c-i+1,\dots,c, c+1,\dots ,c+i$},
for rows $0\leq i < l$, respectively. Let $l_1< l$ be the row with the
largest number
at which this sequence contains only states from $\Sigma\cup\{\border\}$.
Then the states in all these cells are either unavoidable because $l_1=0$ and they are
enforced by the input $w_j$, or they have been derived by replacing  
states not belonging to $\Sigma\cup\{\border\}$. 
Arguing inductively, the latter means that they
have been derived by replacing states that are unavoidable at their positions before.
We conclude that states from $\Sigma\cup\{\border\}$ must not be replaced. So,
whenever the algorithm tries to do this, it halts with a failure message.
On the other hand, if matrix $N_j$ is successfully constructed, then the
current cellular automaton is compatible with the state intervals $I_1$ to $I_{j}$.
This shows the correctness of the algorithm.

It remains to be shown that the algorithm obeys a polynomial time complexity.
The loop at line~\ref{mainloop} is passed through at most $m$ times.
At first, from line~\ref{mainnewmatrixbegin} to line~\ref{mainnewmatrixend},
the matrix $N_j$ is initialized. For each~$j$, this is of size $O(|w_j|\cdot i_j)$.
For all matrices this takes $O((|w_1|+|w_2|+ \cdots +|w_m|)\cdot \range_{\rho})$ steps.
For each such entry within the nested for-loop, line~\ref{fillrowsif} checks
through in the worst case, the size of $s$, which is itself of size  
$O((|w_1|+|w_2|+ \cdots +|w_m|)\cdot \range_{\rho})$. Therefore, lines 
\ref{mainnewmatrixbegin}--\ref{mainnewmatrixend} take 
$O( ((|w_1|+|w_2|+ \cdots +|w_m|)\cdot \range_{\rho})^2)$ time.

In lines~\ref{startlastrow}--\ref{endlastrow}, 
the for-loop is iterating over the last word of an interval, respectively.
In the worst case, it executes the sub-routine \texttt{update}. 

Inside the while loop of \texttt{update}, the nested three for-loops at the
beginning iterate
over each position of the matrices. For each position, the list $tmp$ is
checked, which contains at most $|s|$ elements.
This takes $O(((|w_1|+|w_2|+ \cdots +|w_m|)\cdot \range_{\rho})^2)$ time.
Since in the iterations of the while-loop the number of states is decreasing,
the while-loop is passed through at most $O((|w_1|+|w_2|+ \cdots +|w_m|)\cdot
\range_{\rho})$ times. So, we obtain at most
$O(((|w_1|+|w_2|+ \cdots +|w_m|)\cdot \range_{\rho})^3)$ time, so far.

The next two nested for-loops inside the while-loop scan every entry in $tmp$
and check whether it appears in $s$. For all iterations of the while-loop,
this takes another \mbox{$O(((|w_1|+|w_2|+ \cdots +|w_m|)\cdot \range_{\rho})^3)$} time.

Next, the evaluation of the condition of the for-loop starting at line~\ref{updateresolveconflictsbegin}
takes altogether $O(((|w_1|+|w_2|+ \cdots +|w_m|)\cdot \range_{\rho})^2)$ steps. The
loop is iterated at most $O((|w_1|+|w_2|+ \cdots +|w_m|)\cdot \range_{\rho})$
times.
In the nested for-loop starting at line~\ref{updateresolvecollectconflicts},
all entries of $s$ are compared with a triple of states obtained from the outer
for-loop. So, this part takes at most
$O(((|w_1|+|w_2|+ \cdots +|w_m|)\cdot \range_{\rho})^3)$ time, for all iterations
of the while-loop.
In the following nested if-clause at
line~\ref{updateresolveclassifyconflictsend}
it is checked whether $tmp$ contains duplicate entries, so we obtain
at most
$O(((|w_1|+|w_2|+ \cdots +|w_m|)\cdot \range_{\rho})^3)$ time, for all iterations
of the while-loop.
So, we obtain at most
$O(((|w_1|+|w_2|+ \cdots +|w_m|)\cdot \range_{\rho})^3)$ time, so far.

Thus, in total, the sub-routine \texttt{update} 
takes $O(((|w_1|+|w_2|+ \cdots +|w_m|)\cdot \range_{\rho})^4)$ time.
Because it is called at most $O( ((|w_1|+|w_2|+ \cdots +|w_m|)\range_{\rho})^2)$ times,
the entire algorithm executes in
\mbox{$O( ((|w_1|+|w_2|+ \cdots +|w_m|)\range_{\rho})^6)$} time.
Since the distances are given in unary, it operates in~$O(n^6)$ time where $n$ is the size of the input.

This concludes the proof of the proposition for $\ca$. The adaptation to $\oca$
is straightforward.
\end{proof}

\begin{example}\label{exa:it-inv-prob}
\begin{sloppypar}
Figure~\ref{fig:exa-it-inv-prob} shows the beginning of a successful solving 
of a $\ca$ Interval Inference Problem.
The state intervals are $(abcabc, abcccc, 3)$ and $(abcabcc, ccccccc, 5)$, so
we have \mbox{$\Sigma=\{a=1, b=2, c=3\}$.}
The dynamic changes in the matrices and the values of $s[a],s[b],s[c]$ are
shown.
\end{sloppypar}

The matrix in the first row is obtained for $j=1$ in the main part of the
algorithm up to line~\ref{endfillrows}. 
Notice that already in this matrix, some states do and must repeat, such as the entry 5 in row 1, because both positions depend on
$(a,b,c)$ in the row above it.
Then the last row is filled, where the
states $15$ to $20$ are replaced. Since they are all different and do appear only in the last row of
$N_1$, lines~\ref{startlastrow}--~\ref{endlastrow} do not reject but 
call \texttt{update} with list $upd = [(i,b_{1,c})]$, and the while-loop of
\texttt{update} is run through once, which does not cause further replacements
and updates of $s$. The result is the matrix on the left of the second row.

Next, $j$ is set to $2$ and $N_2$ is pre-filled with states according to the
already defined transition rules and by using new states, where necessary.
The result is the matrix on the right of the second row, and the current 
values of $s[a],s[b],s[c]$ are shown below it. Now, the replacement of the
states in the last row starts (line~\ref{startlastrow}). The first letter
to be replaced is the $9$ (circled for visualization). So, list $upd$ gets the
pair~$(9,c)$ and \texttt{update} is called. Sub-routine \texttt{update} first
replaces all occurrences of $9$ by $c$ in~$N_1$ and~$N_2$. Notice that this causes possible
changes even to $N_1$ as the $9$ changes to $c$ in row 2. At the same time,
the only triple $(\border,4,5)$ mapping to $9$ is moved in $s$ to map to $c$ 
(lines~\ref{updateallcellsbegin}--~\ref{updateonecellsend}). This gives the
matrices in the third row. In
lines~\ref{updatetransitionsbegin}--~\ref{updatetransitionsend}
all triples including letter $9$ are updated to include letter $c$; 
that is, $(\border,9,10)$ becomes $(\border,c,10)$ and $(9,10,11)$ becomes $(c,10,11)$.
The
resulting values of $s[a],s[b],s[c]$ are shown below the matrices. Since the
rest of \texttt{update} does not apply, this ends the replacement of the first
letter of the last row. 

\begin{figure}[!b]
\centering
\includegraphics[scale=.75]{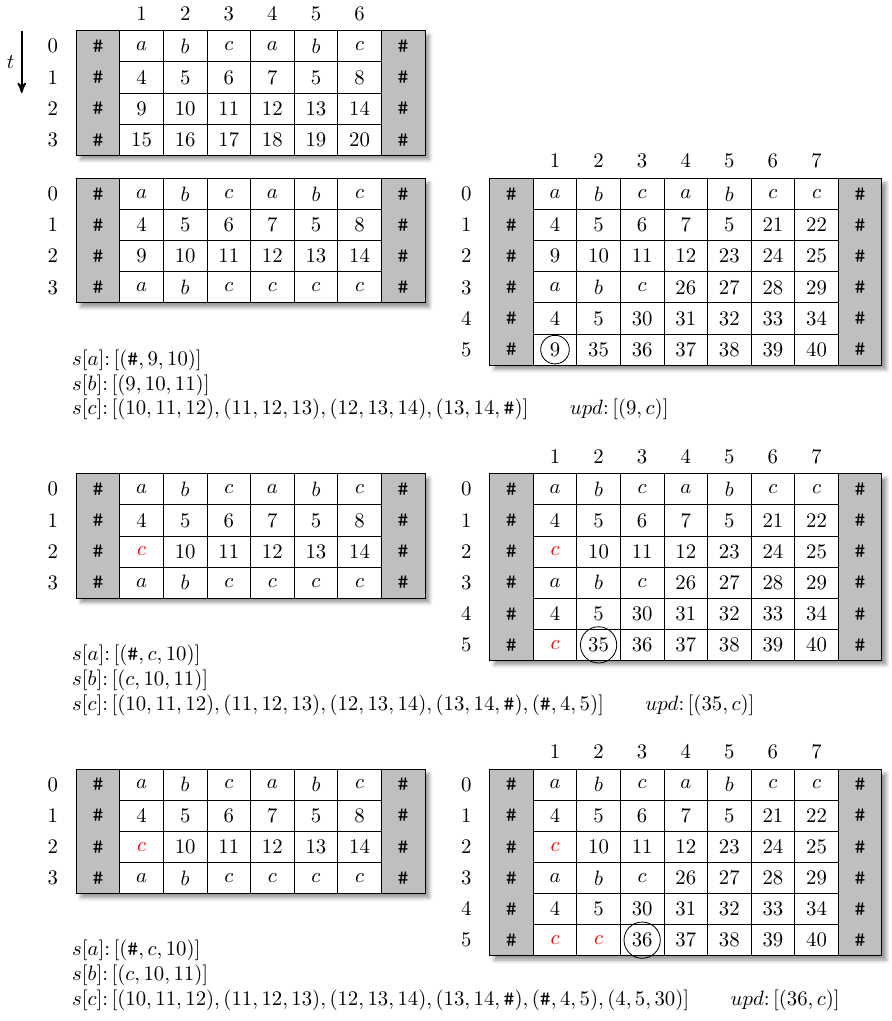}
\caption{Example~\ref{exa:it-inv-prob} for the beginning of a successful solving of a $\ca$ Interval Inference Problem. The state
  intervals are $(abcabc, abcccc, 3)$ and $(abcabcc, ccccccc, 5)$.}
\label{fig:exa-it-inv-prob}
\end{figure}

The second state to be replaced is the $35$ (circled for visualization). So, list $upd$ gets the
pair $(35,c)$ and \texttt{update} is called.
Since $35$ and the remaining letters in the last row are all different and do 
appear only in the last row of $N_2$, lines~\ref{startlastrow}--~\ref{endlastrow} do not reject but 
call \texttt{update} with list $upd = [(i,b_{2,c})]$, and the while-loop of
\texttt{update} is run through once, respectively, which does not cause further replacements
and updates of $s$. The result after the replacement are the matrices in the
fourth row. Again, the resulting values of $s[a],s[b],s[c]$ after the
replacement of $35$ are shown below the matrices.
\eoe
\end{example}

Notice in this example, the list $s$ that is output contains positions that are empty, and therefore these are states that were created and then eliminated. There are 53 positions of the two matrices after the first rows in which a different state is possible, but only
40 states were initially created (the indices of~$s$ go to~40). Furthermore,
the final $\ca$ only contains 28 states, and therefore 12 more were eliminated.

\subsection{The Partial Interval Inference Problem}

So far, we considered two problems where the input consists of a state
interval set and a transition function between two extremes: for the
Interval Verification Problem, the given transition is complete, while the
``given'' transition function for the Interval Inference Problem is empty.
This raises the natural question for a problem in between, that is formulated into
the Partial Interval Inference Problem, where a transition function is given
partially. It turns out that the problem can be treated by shifting back and 
forth between the two extremes.

\begin{problem}{Partial $\ca$ (resp.\ $\oca$) Interval Inference Problem.}
Given a state interval set $\rho = (w_1, w_1', i_1), \ldots, (w_m, w_m', i_m)$
and a partial transition function~$\delta$, is there a $\ca$ 
(resp.\ $\oca$) with local transition function $\delta' \supseteq \delta$ that
is compatible with~$\rho$?
If yes then construct it.
\end{problem}

\begin{theorem}\label{theo:partial-interval-inference-prob}
The Partial $\ca$ (resp.\ $\oca$) Interval Inference Problem is 
\mbox{$\P$-complete.} 
\end{theorem}

The argument showing the $\P$-hardness in the first part of the proof 
of Theorem~\ref{theo:partial-interval-inference-prob} utilizes the
possibility to provide a degenerated partial transition function;
that is, a complete transition function, in which case the
Partial Interval Inference Problem boils down to the Interval
Verification Problem. Do we need a ``partial'' transition function
in the input for this purpose? The answer is no, seen as follows.

\begin{theorem}\label{theo:interval-inference-prob-hardness}
The $\ca$ (resp.\ $\oca$) Interval Inference Problem is 
$\P$-complete. 
\end{theorem}

\section{Concluding Remarks}\label{sec:conclusions}

To classify the computational complexity of a problem on a finer scale than
measuring it as a function of its size (i.e., number of bits in the input)
the so-called parameterized complexity is studied~\cite{downey:1999:pc,gurevich:1984:snphpgat}.
Parameterized complexity considers additional parameters of the input instance.
An often cited example is the vertex cover problem: ``Is there a vertex cover of size
$k$''? This problem has the number $k$ of vertices in the cover as a natural parameter, 
which is independent of the size of the input graph. The vertex cover
problem is $\NP$-complete if the input parameter $k$ is not fixed. However,
the problem can be solved by an algorithm that is exponential only in $k$
while it is polynomial in the size of the input:~\mbox{$O(2^k\cdot n)$.} 
Such an algorithm is called fixed-parameter tractable (FPT), 
because the problem can be solved efficiently in polynomial time
for constant values of the fixed parameter. 
More formally, a \emph{parameterized problem} is a language $L \subseteq \Sigma^*\times\N$, where
$\Sigma$ is a finite, fixed alphabet. 
For instance in $(x, k)\in \Sigma^*\times \mathbb{N}$, we refer to $k$ as the parameter.
A parameterized problem $L$ is \emph{fixed-parameter tractable (FPT)} if there exists
a computable function $f$, a constant $c$, and algorithm which correctly
decides whether $(x, k)\in L$ is time bounded by \mbox{$f(k) \cdot |(x, k)|^c$.}

\begin{sloppypar}
We introduced the notion of state intervals $(w,w',i)$ such that the 
distance~$i$ is given in unary, essentially to provide one input bit for each
row of space-time diagrams. However, this condition can be given up if we 
consider the parameterized complexity of our problems.
For example, the parameterized version of the Interval Verification Problem
is ``Given a state interval set \mbox{$\rho = (w_1, w_1', i_1), \ldots, (w_m, w_m',
i_m)$} with $\range_\rho \leq k$, and a $\ca$ (resp.\ $\oca$) $M$, is $M$ compatible with $\rho$?''.
So, we consider the range of the state interval set as parameter. 
Assume that the problem belongs to $\P$ if the distances of the state
intervals are given in unary; then,  if the distances of the state intervals are given in
binary, we can use the ``harmless'' function $f(k)=k$ to obtain a polynomially
time bounded solution.
\end{sloppypar}


\begin{thebibliography}{10}
\providecommand{\bibitemdeclare}[2]{}
\providecommand{\surnamestart}{}
\providecommand{\surnameend}{}
\providecommand{\urlprefix}{Available at }
\providecommand{\url}[1]{\texttt{#1}}
\providecommand{\href}[2]{\texttt{#2}}
\providecommand{\urlalt}[2]{\href{#1}{#2}}
\providecommand{\doi}[1]{doi:\urlalt{https://doi.org/#1}{#1}}
\providecommand{\eprint}[1]{arXiv:\urlalt{https://arxiv.org/abs/#1}{#1}}
\providecommand{\bibinfo}[2]{#2}

\bibitemdeclare{book}{Delorme:1999:capm}
\bibitem{Delorme:1999:capm}
\bibinfo{editor}{Marianne \surnamestart Delorme\surnameend} \&
  \bibinfo{editor}{Jacques \surnamestart Mazoyer\surnameend}, editors
  (\bibinfo{year}{1999}): \emph{\bibinfo{title}{Cellular Automata -- a Parallel
  Model}}.
\newblock \bibinfo{publisher}{Kluwer Academic Publishers}.

\bibitemdeclare{book}{downey:1999:pc}
\bibitem{downey:1999:pc}
\bibinfo{author}{Rodney~G. \surnamestart Downey\surnameend} \&
  \bibinfo{author}{Michael~R. \surnamestart Fellows\surnameend}
  (\bibinfo{year}{1999}): \emph{\bibinfo{title}{Parameterized Complexity}}.
\newblock \bibinfo{series}{Monographs in Computer Science},
  \bibinfo{publisher}{Springer}, \doi{10.1007/978-1-4612-0515-9}.

\bibitemdeclare{article}{ComplexityInference}
\bibitem{ComplexityInference}
\bibinfo{author}{Christopher \surnamestart Duffy\surnameend},
  \bibinfo{author}{Sam \surnamestart Hillis\surnameend}, \bibinfo{author}{Umer
  \surnamestart Khan\surnameend}, \bibinfo{author}{Ian \surnamestart
  McQuillan\surnameend} \& \bibinfo{author}{Sonja~Linghui \surnamestart
  Shan\surnameend} (\bibinfo{year}{2025}): \emph{\bibinfo{title}{Inductive
  inference of {Lindenmayer} systems: algorithms and computational
  complexity}}.
\newblock {\slshape \bibinfo{journal}{Natural Computing}} \bibinfo{volume}{24},
  pp. \bibinfo{pages}{591--601}, \doi{10.1007/s11047-025-10024-x}.

\bibitemdeclare{article}{goldschlager:1977:mpcvppcomplete}
\bibitem{goldschlager:1977:mpcvppcomplete}
\bibinfo{author}{Leslie~M. \surnamestart Goldschlager\surnameend}
  (\bibinfo{year}{1977}): \emph{\bibinfo{title}{The monotone and planar circuit
  value problems are log space complete for {P}}}.
\newblock {\slshape \bibinfo{journal}{{SIGACT} News}} \bibinfo{volume}{9}, pp.
  \bibinfo{pages}{25--29}, \doi{10.1145/1008354.1008356}.

\bibitemdeclare{article}{gurevich:1984:snphpgat}
\bibitem{gurevich:1984:snphpgat}
\bibinfo{author}{Yuri \surnamestart Gurevich\surnameend},
  \bibinfo{author}{Larry~J. \surnamestart Stockmeyer\surnameend} \&
  \bibinfo{author}{Uzi \surnamestart Vishkin\surnameend}
  (\bibinfo{year}{1984}): \emph{\bibinfo{title}{Solving NP-Hard problems on
  graphs that are almost trees and an application to facility location
  problems}}.
\newblock {\slshape \bibinfo{journal}{J. {ACM}}} \bibinfo{volume}{31}, pp.
  \bibinfo{pages}{459--473}, \doi{10.1145/828.322439}.

\bibitemdeclare{book}{hermanrozenberg}
\bibitem{hermanrozenberg}
\bibinfo{author}{Gabor~T. \surnamestart Herman\surnameend} \&
  \bibinfo{author}{Grzegorz \surnamestart Rozenberg\surnameend}
  (\bibinfo{year}{1975}): \emph{\bibinfo{title}{Developmental Systems and
  Languages}}.
\newblock \bibinfo{publisher}{North-Holland Publishing Company},
  \bibinfo{address}{Oxford}.

\bibitemdeclare{book}{delaHiguera2010}
\bibitem{delaHiguera2010}
\bibinfo{author}{Colin \surnamestart de~la Higuera\surnameend}
  (\bibinfo{year}{2010}): \emph{\bibinfo{title}{Grammatical Inference: Learning
  Automata and Grammars}}.
\newblock \bibinfo{publisher}{Cambridge University Press},
  \doi{10.1017/CBO9781139194655}.

\bibitemdeclare{article}{Kari:2005:tocas}
\bibitem{Kari:2005:tocas}
\bibinfo{author}{Jarkko \surnamestart Kari\surnameend} (\bibinfo{year}{2005}):
  \emph{\bibinfo{title}{Theory of cellular automata: a survey}}.
\newblock {\slshape \bibinfo{journal}{Theor. Comput. Sci.}}
  \bibinfo{volume}{334}(\bibinfo{number}{1-3}), pp. \bibinfo{pages}{3--33},
  \doi{10.1016/j.tcs.2004.11.021}.

\bibitemdeclare{incollection}{kutrib:2008:ca-cpv}
\bibitem{kutrib:2008:ca-cpv}
\bibinfo{author}{Martin \surnamestart Kutrib\surnameend}
  (\bibinfo{year}{2008}): \emph{\bibinfo{title}{Cellular automata -- a
  computational point of view}}.
\newblock In \bibinfo{editor}{G.~\surnamestart Bel-Enguix\surnameend},
  \bibinfo{editor}{M.~D. \surnamestart Jim{\'e}nez-L{\'o}pez\surnameend} \&
  \bibinfo{editor}{C.~\surnamestart Mart{\'i}n-Vide\surnameend}, editors:
  {\slshape \bibinfo{booktitle}{New Developments in Formal Languages and
  Applications}}, chapter~\bibinfo{chapter}{6}, \bibinfo{publisher}{Springer},
  pp. \bibinfo{pages}{183--227}, \doi{10.1007/978-3-540-78291-9\_6}.

\bibitemdeclare{incollection}{kutrib:2009:calt}
\bibitem{kutrib:2009:calt}
\bibinfo{author}{Martin \surnamestart Kutrib\surnameend}
  (\bibinfo{year}{2009}): \emph{\bibinfo{title}{Cellular automata and language
  theory}}.
\newblock In \bibinfo{editor}{R.~\surnamestart Meyers\surnameend}, editor:
  {\slshape \bibinfo{booktitle}{Encyclopedia of Complexity and System
  Science}}, \bibinfo{publisher}{Springer}, pp. \bibinfo{pages}{800--823},
  \doi{10.1007/978-0-387-30440-3\_54}.

\bibitemdeclare{incollection}{kutrib:2018:cadcad}
\bibitem{kutrib:2018:cadcad}
\bibinfo{author}{Martin \surnamestart Kutrib\surnameend} \&
  \bibinfo{author}{Andreas \surnamestart Malcher\surnameend}
  (\bibinfo{year}{2018}): \emph{\bibinfo{title}{Cellular automata:
  descriptional complexity and decidability}}.
\newblock In \bibinfo{editor}{Andrew \surnamestart Adamatzky\surnameend},
  editor: {\slshape \bibinfo{booktitle}{Reversibility and Universality}},
  {\slshape \bibinfo{series}{Emergence, Complexity and
  Computation}}~\bibinfo{volume}{30}, \bibinfo{publisher}{Springer}, pp.
  \bibinfo{pages}{129--168}, \doi{10.1007/978-3-319-73216-9\_6}.

\bibitemdeclare{inproceedings}{UCNC2018}
\bibitem{UCNC2018}
\bibinfo{author}{Ian \surnamestart McQuillan\surnameend},
  \bibinfo{author}{Jason \surnamestart Bernard\surnameend} \&
  \bibinfo{author}{Przemyslaw \surnamestart Prusinkiewicz\surnameend}
  (\bibinfo{year}{2018}): \emph{\bibinfo{title}{Algorithms for inferring
  context-sensitive {L}-systems}}.
\newblock In \bibinfo{editor}{S.~\surnamestart Stepney\surnameend} \&
  \bibinfo{editor}{S.~\surnamestart Verlan\surnameend}, editors: {\slshape
  \bibinfo{booktitle}{Proceedings of the 17th International Conference on
  Unconventional Computation and Natural Computation, UCNC 2018}}, {\slshape
  \bibinfo{series}{LNCS}} \bibinfo{volume}{10867}, pp.
  \bibinfo{pages}{117--130}, \doi{10.1007/978-3-319-92435-9\_9}.

\bibitemdeclare{inproceedings}{neary:2006:pccar110}
\bibitem{neary:2006:pccar110}
\bibinfo{author}{Turlough \surnamestart Neary\surnameend} \&
  \bibinfo{author}{Damien \surnamestart Woods\surnameend}
  (\bibinfo{year}{2006}): \emph{\bibinfo{title}{P-completeness of Cellular
  Automaton Rule 110}}.
\newblock In \bibinfo{editor}{Michele \surnamestart Bugliesi\surnameend},
  \bibinfo{editor}{Bart \surnamestart Preneel\surnameend},
  \bibinfo{editor}{Vladimiro \surnamestart Sassone\surnameend} \&
  \bibinfo{editor}{Ingo \surnamestart Wegener\surnameend}, editors: {\slshape
  \bibinfo{booktitle}{International Colloquium on Automata, Languages and
  Programming (ICALP 2006)}}, \bibinfo{series}{LNCS},
  \bibinfo{publisher}{Springer}, pp. \bibinfo{pages}{132--143},
  \doi{10.1007/11786986\_13}.

\bibitemdeclare{book}{Rozenberg:2012:HandbookNC}
\bibitem{Rozenberg:2012:HandbookNC}
\bibinfo{editor}{Grzegorz \surnamestart Rozenberg\surnameend},
  \bibinfo{editor}{Thomas \surnamestart B{\"{a}}ck\surnameend} \&
  \bibinfo{editor}{Joost~N. \surnamestart Kok\surnameend}, editors
  (\bibinfo{year}{2012}): \emph{\bibinfo{title}{Handbook of Natural
  Computing}}.
\newblock \bibinfo{publisher}{Springer}, \doi{10.1007/978-3-540-92910-9}.

\end{thebibliography}
\end{document}